\documentclass[a4paper]{article}

\usepackage[IL2]{fontenc}
\usepackage[letterpaper,margin=1.00in]{geometry}
\usepackage{amsmath, amssymb, amsthm, amsfonts}
\usepackage{bbm}
\usepackage{amscd}
\usepackage{mathrsfs}

\usepackage{comment} 
\usepackage{ifthen}
\usepackage{tikz}
\usetikzlibrary{positioning,decorations.pathreplacing}
\usepackage{ bbold }
\usepackage{appendix}
\usepackage{graphicx}
\usepackage{color}
\usepackage{algorithm}
\usepackage[noend]{algpseudocode}
\usepackage{epstopdf}
\usepackage{wrapfig}
\usepackage{paralist}
\usepackage{wasysym}
\usepackage[textsize=tiny]{todonotes}

\usepackage{pdflscape} 

\usepackage{framed}
\usepackage[framemethod=tikz]{mdframed}
\usepackage[bottom]{footmisc}
\usepackage{enumitem}
\setitemize{noitemsep,topsep=3pt,parsep=3pt,partopsep=3pt}
\usepackage[font=small]{caption}
\usepackage{xspace}

\usepackage{thmtools} 
\usepackage{thm-restate} 

\usepackage{hyperref}
\hypersetup{
    unicode=false,          
    colorlinks=true,        
    linkcolor=red,          
    citecolor=darkgreen,        
    filecolor=magenta,      
    urlcolor=cyan           
}

\newtheorem{theorem}{Theorem}[section]
\newtheorem{lemma}[theorem]{Lemma}
\newtheorem{meta-theorem}[theorem]{Meta-Theorem}

\newtheorem{remark}[theorem]{Remark}

\newtheorem{definition}[theorem]{Definition}

\usepackage[capitalize, nameinlink]{cleveref}
\crefname{theorem}{Theorem}{Theorems}
\crefname{proposition}{Proposition}{Propositions}
\crefname{observation}{Observation}{Observations}
\crefname{lemma}{Lemma}{Lemmas}
\crefname{claim}{Claim}{Claims}
\crefname{problem}{Problem}{Problems}
\crefname{conjecture}{Conjecture}{Conjectures}
\crefname{question}{Question}{Questions}
\crefname{example}{Example}{Examples}
\crefname{fact}{Fact}{Facts}

\definecolor{darkgreen}{rgb}{0,0.5,0}

\algnewcommand\algorithmicswitch{\textbf{switch}}
\algnewcommand\algorithmiccase{\textbf{case}}

\algdef{SE}[SWITCH]{Switch}{EndSwitch}[1]{\algorithmicswitch\ #1\ \algorithmicdo}{\algorithmicend\ \algorithmicswitch}%
\algdef{SE}[CASE]{Case}{EndCase}[1]{\algorithmiccase\ #1}{\algorithmicend\ \algorithmiccase}%
\algtext*{EndSwitch}%
\algtext*{EndCase}%

\newcommand{\eps}{\varepsilon}

\renewcommand{\P}{\textrm{P}}

\newcommand{\E}{\textrm{E}}

\newcommand{\R}{\mathbb{R}}

\renewcommand{\paragraph}[1]{\vspace{0.15cm}\noindent {\bf #1}:}

\newcommand{\FullOrShort}{full}
\ifthenelse{\equal{\FullOrShort}{full}}{
  
  \newcommand{\fullOnly}[1]{#1}
  \newcommand{\shortOnly}[1]{}

  }{

    \newcommand{\fullOnly}[1]{}
    \newcommand{\IncludePictures}[1]{}
   
  }

\renewcommand{\phi}{\varphi}

\usepackage{listings}
\newcommand{\start}{i_{2\ell}}
\newcommand{\test}{i_{\ell}}

\title{Noisy k-means++ Revisited} 

\author{
  Christoph Grunau \\
  \small{ETH Zurich}\\
  \small{cgrunau@inf.ethz.ch}
  \and
  Ahmet Alper Özüdoğru \\
  \small{ETH Zurich}\\
  \small{oahmet@student.ethz.ch}
  \and  
Václav Rozhoň\\ 
\small{ETH Zurich} \\ 
\small{rozhonv@ethz.ch}
}

\begin{document}

\maketitle

\begin{abstract}
    The $k$-means++ algorithm by Arthur and Vassilvitskii [SODA 2007] is a classical and time-tested algorithm for the $k$-means problem. While being very practical, the algorithm also has good theoretical guarantees: its solution is $O(\log k)$-approximate, in expectation.  
    
    In a recent work, Bhattacharya, Eube, Roglin, and Schmidt [ESA 2020] considered the following question: does the algorithm retain its guarantees if we allow for a slight adversarial noise in the sampling probability distributions used by the algorithm? This is motivated e.g. by the fact that computations with real numbers in $k$-means++ implementations are inexact. 
    Surprisingly, the analysis under this scenario gets substantially more difficult and the authors were able to prove only a weaker approximation guarantee of $O(\log^2 k)$. In this paper, we close the gap by providing a tight, $O(\log k)$-approximate guarantee for the $k$-means++ algorithm with noise. 
\end{abstract}

\section{Introduction}

The $k$-means problem is a classical problem in computer science: given a \emph{point set} $X \subseteq \R^d$ consisting of $n$ points  and a parameter $k$,
we are asked to return a set of $k$ \emph{clusters} with corresponding cluster centers $C \subseteq \R^d$ so as to minimize the sum of the squared distances of points of $X$ with respect to their closest cluster center in $C$.
Formally, we are asked to minimize the function $\phi(X, C)$ defined by $\phi(x, C) = \min_{c \in C} || x - c||^2$ for a single point $x$ and as $\phi(X, C) = \sum_{x \in X} \phi(x, C)$ for a set of points.  

There exists some fixed constant $c > 1$ such that it is NP-hard to find a $c$-approximate solution to the $k$-means objective \cite{aloise2009np,awasthi2015hardness}.
On the other hand, a substantial amount of work has been devoted to finding polynomial time algorithms with a good approximation guarantee, with the currently best approximation ratio being 5.912 \cite{cohenaddad2022best_apx}.
On the practical side, the celebrated clustering algorithm $k$-means++ by Arthur and Vassilvitskii \cite{arthur2007k} is one of the classical algorithms for the $k$-means problem. Due to its simplicity, it is widely used in practice, for example in the well-known Python Scikit-learn library \cite{scikit-learn}. It is also very appealing from the theoretical perspective, as it returns a solution that is $O(\log k)$-approximate, in expectation. 

The $k$-means++ algorithm (\cref{alg:kmpp_eps_adversary} with $\eps = 0$) is indeed very simple: we sample $C \subseteq X$ in $k$ steps. The first center is taken as a uniformly random point of $X$. To get each subsequent center, we always first compute the current costs $\phi(x, C_{i})$ for each $x \in X$; then we sample each point of $X$ as the next center with probability proportional to $\phi(x, C_i)$. 

In \cite{bhattacharya2019noisy}, the authors made an intriguing observation: the classical analysis of the algorithm by Arthur and Vassilvitski \cite{arthur2007k} fails to work if we allow small errors in the sampling probabilities. That is, consider \cref{alg:kmpp_eps_adversary}: this is the $k$-means++ algorithm, however, with an additional small positive parameter $\eps$. 
In every step, before we sample, we allow an adversary to perturb the sampling distribution such that the multiplicative change of each probability is within $1 \pm \eps$ of its original value.

\begin{algorithm}[h]
	\caption{$(1+\eps)$-noisy $k$-means++}
	\label{alg:kmpp_eps_adversary}
	Input: $X$, $k$, $0 \leq \eps < 1/2$
	\begin{algorithmic}[1]
		\State Sample $x \in X$ w.p. in $\left[(1 - \eps) \cdot \frac{1}{n}, (1 + \eps) \cdot \frac{1}{n}\right]$, set $C_1 = \{ x \}$. 
		\For{$i \leftarrow 0, 1, \dots, k-1$}
		\State Sample $x \in X$ w.p. in $\left[(1 - \eps) \cdot \frac{\phi(x, C_{i})}{\phi(X, C_{i})}, (1 + \eps) \cdot \frac{\phi(x, C_{i})}{\phi(X,C_{i})}\right]$ and set $C_{i+1} = C_{i} \cup \{x\}$.
		\EndFor
	\Return $C := C_k$
	\end{algorithmic}
\end{algorithm}

Does the noisy $k$-means++ algorithm retain the original guarantees? This question is natural since in every implementation, there are small numerical errors associated with the distance computations made by \cref{alg:kmpp_eps_adversary}. It would be shocking if these errors could substantially affect the quality of the algorithm's output! From a more theoretical perspective, the authors of \cite{bhattacharya2019noisy} considered this problem as a first step towards understanding other questions related to the $k$-means++ algorithm, in particular the analysis of the \emph{greedy} variant of $k$-means++, a related algorithm later analyzed in \cite{grunau_ozudogru_rozhon_tetek2022nearly}. 

Going back to noisy $k$-means++, the authors of \cite{bhattacharya2019noisy} proved that \cref{alg:kmpp_eps_adversary} remains $O(\log^2 k)$-approximate even for small constant $\eps$ (think e.g. $\eps = 0.01$). In this paper, we improve their analysis to recover the tight $O(\log k)$-approximation guarantee. That is, we show that the adversarial noise worsens the approximation guarantee by at most a constant multiplicative factor.

\begin{theorem}
\label{thm:main}
\cref{alg:kmpp_eps_adversary} is $O(\log k)$-approximate, in expectation. 
\end{theorem}

\begin{remark}
It would be interesting to see an analysis of the approximation ratio of \cref{alg:kmpp_eps_adversary} that would be within a $1+O(\eps)$-factor of the classical $k$-means++ analysis from \cite{arthur2007k}, or a counterexample showing this is not possible. In our analysis, we lose a very large constant factor even for very small $\eps$. 
\end{remark}

\paragraph{Related Work}
There is a lot of work related to the $k$-means++ algorithm, both improving the algorithm or its analysis \cite{lattanzi2019better,choo2020kmeans,aggarwal2009adaptive,wei2016constant,makarychev_reddy_shan2020improved,bhattacharya2019noisy,grunau_ozudogru_rozhon_tetek2022nearly} and adapting it to other setups \cite{bahmani2012scalable,bachem2016approximate, rozhovn2020simple,makarychev_reddy_shan2020improved,bachem2016fast,bachem2017distributed,bhaskara2019kmeans++_with_outliers_and_penalties,grunau_rozhon2020adapting_kmeans_to_outliers}.

\paragraph{Acknowledgements} We would like to thank Mohsen Ghaffari for many helpful comments. 

\section{Reduction to a Sampling Game}
\label{sec:game}

To analyze \cref{alg:kmpp_eps_adversary}, the authors of \cite{bhattacharya2019noisy} follow the proof of \cite{arthur2007k} (more precisely, they follow the proof from \cite{Dasgupta}) and show that most arguments of that proof, in fact, work even in the adversarial noise scenario. The part of the proof that does not generalize from $\eps = 0$ to $\eps > 0$ can be distilled into a simple sampling process that we analyze in this paper. We next describe this process and state its relation to the analysis of noisy $k$-means++ (cf. the discussion on page 15 of \cite{bhattacharya2019noisy}). 

\begin{definition}[$(1+\eps)$-adversarial sampling process]
\label{def:process_eps_adversary}
Let $0 < \eps < 1/2$. 
We define the \emph{$(1+\eps)$-adversarial sampling process} as follows. 
At the beginning, there is a set $E_0$ of $k$ elements where each element $e \in E_0$ has some nonnegative weight $w_0(e)$. 
The process has $k$ rounds where in each round, we form the new set $E_{i+1}$ from $E_i$ as follows:
\begin{enumerate}
    \item We define the distribution $D_i$ over $E_i$ where the probability of selecting $e \in E_i$ is defined as $w_i(e) / \sum_{e \in E_i} w_i(e)$. Next, an adversary chooses an arbitrary distribution $D^{\eps}_i$ over $E_i$ that satisfies for any $e \in E_i$ that
    \begin{equation}
    \label{eq:adversarial_dist}
    (1-\eps)\P_{D_i}(e) \leq \P_{D^\eps_i}(e) \leq (1+\eps)\P_{D_i}(e).
    \end{equation}
    We sample an element $e_{i+1} \in E_i$ according to $D_i^\eps$ and set $E_{i+1} = E_i \setminus \{e_{i+1}\}$. 
    \item Next, an adversary chooses a new weight function $w_{i+1}(e)$ for every element $e \in E_{i+1}$ as an arbitrary function that satisfies
    \[
    0 \le w_{i+1}(e) \le w_i(e). 
    \]
\end{enumerate}
\end{definition}

We will be interested in the expected average weight of an element after some number of steps in this process, that is, we need to understand the value of $\E\left[\frac{\sum_{e \in E_i} w_i(e)}{k - i}\right]$ for $0 \le i < k$. If $\eps = 0$, one can prove that 
\begin{equation}
\label{eq:easy}
\E\left[\frac{\sum_{e \in E_i} w_i(e)}{k - i}\right] \le \frac{\sum_{e \in E_{i-1}} w_{i-1}(e)}{k - (i-1)}
\end{equation}
where the randomness is over the sampling in the $i$-th step (we always regard the adversary as fixed in advance). Why is \cref{eq:easy} true? The inequality would clearly hold with equality if the distribution $D_i$ were a uniform one and there was no adversary; we in fact give larger sampling probabilities to heavier elements in $D_i$ and, moreover, the adversary can lower the weights arbitrarily after we sample, but both of these operations can make the left-hand side of \cref{eq:easy} only smaller. 

However, this monotonic behavior is no longer true for $\eps > 0$. The question that needs to be analyzed as a part of the analysis of noisy $k$-means++ is whether the adversarial choices can make the average size of an element drift so that in the end the left-hand side of \cref{eq:easy} is substantially larger than $\sum_{e \in E_0} w_0(e) / k$. More precisely, we will need to bound the following quantity that we call the adversarial advantage.

\begin{definition}[Adversarial advantage]
We say that the adversarial advantage is at most some function $f$ if the following conclusion holds: Consider a $(1+\eps)$-adversarial sampling process on $k$ elements for any $0 < \eps < \frac12$, any starting set $E_0$, and any adversary. For any $0 \le i < k$, we have
\begin{equation}
\label{eq:advantage}
\E\left[\frac{\sum_{e \in E_i} w_i(e)}{k - i}\right] 
\le f(k) \cdot  \frac{\sum_{e \in E_0} w_0(e)}{k}. 
\end{equation}
\end{definition}

Although we require the inequality \cref{eq:advantage} to hold for all $i$, note that for all $0 \le i \le (1-\delta)k$ we can choose $f(k) = 1/\delta$ in \cref{eq:advantage} and it will be satisfied for those values of $i$ simply because $\sum_{e \in E_i} w_i(e) \le \sum_{e \in E_0} w_0(e)$ is true deterministically. Thus, intuitively, $i = k-1$ is the hardest case.

In \cite{bhattacharya2019noisy}, the authors proved that if we adapt the analysis of $k$-means++ to the noisy $k$-means++, it only picks up the multiplicative factor of $f(k)$. That is, analyzing the $(1+\eps)$-adversarial sampling process is enough to get an upper bound for noisy $k$-means++. The following theorem is proven in \cite{bhattacharya2019noisy} (it is proven only for $f(k) = O(\log k)$, but it directly generalizes to any $f(k)$). 

\begin{theorem}[Theorem 2 in \cite{bhattacharya2019noisy}]
\label{thm:reduction}
For any $0 < \eps < 1/2$, $(1+\eps)$-noisy $k$-means++ is $O(f(k) \cdot \log k)$-approximate, in expectation. 
\end{theorem}

In Lemma 10 of \cite{bhattacharya2019noisy}, the authors prove that $f(k) = O(\log k)$. The reason for this is that if an element $e \in E_0$ is $\Theta(\log k)$ times larger than the average size of an element of $E_0$, it will be sampled in the first $k/2$ steps of the process with probability $1 - 1/k^{O(1)}$. Thus, the contribution of elements $\Omega(\log k)$ larger than the average to the left-hand side of \cref{eq:advantage} is negligible even for $i = k-1$. Hence, $f(k) = O(\log k)$. 

\begin{lemma}[Lemma 10 in \cite{bhattacharya2019noisy}]
The adversarial advantage is at most $O(\log k)$. 
\end{lemma}

Our technical contribution is to show that the adversarial advantage is bounded by $O(1)$. 

\begin{lemma}
\label{lem:main}
The adversarial advantage is at most $O(1)$. 
\end{lemma}

\cref{thm:main} then follows from \cref{thm:reduction} and \cref{lem:main}. 

\section{Analysis of the Sampling Process}
\label{sec:analysis}

This section is devoted to the proof of \cref{lem:main}. 
We view the adversary as a function fixed at the beginning of the argument. We start by normalizing the starting weights $w_0$ so that the average at the beginning is one, i.e., from now on we assume that $(\sum_{e \in E_0} w_0(e)) / k = 1$. For every $E \subseteq E_i$, we define $w_i(E) = \sum_{e \in E} w_i(e)$ and similarly $P_{D^\eps_i}(E) = \sum_{e \in E} P_{D^\eps_i}(e)$. In every step $i$, we consider the partition $E_i = B_i \sqcup M_i \sqcup S_i$ where $e \in E_i$ is in
\begin{enumerate}
    \item the big set $B_i$ iff $w_i(e) \geq 80$,
    \item the medium set $M_i$ iff $2 < w_i(e) < 80$ and
    \item the small set $S_i$ iff $w_i(e) \leq 2$. 
\end{enumerate}
The main idea of the analysis is to show that $w_i(B_i) = O(|S_i|)$, and thus $\frac{w_i(E_i)}{k-i} = \frac{O(|S_i|)}{|S_i| + |M_i| + |B_i|} = O(1)$, with probability $1 - e^{-\Omega(|S_i|)}$. This turns out (see the proof of \cref{lem:main}) that this is sufficient to show that the adversarial advantage is $O(1)$, i.e., that $\E\left[\frac{w_i(E_i)}{k-i}\right] = O(1)$. 

Roughly speaking, we call an iteration with $\ell$ small elements bad, if the total weight of the big elements is
greater than $4\ell$, which intuitively means the average drifted way above 1. In general we use the number of the small elements as our main way to refer to the iterations. Then in \cref{lem:approx} we denote with $\ell_{max}$ the number of small elements at the first bad iteration. Using that the previous iterations were good, and $w_{\start}(B_{\start}) \leq 8\ell$ for the bad iterations (\cref{def:bad}), we provide an upper bound on the average element size for the following iterations. Even though this bound is depending on the number of the small elements $\ell$, we show in \cref{lem:bad} that an iteration is bad with probability at most $e^{-\frac{\ell}{40}}$, which is enough to show the constant average in expectation.

The following definition is crucial for our analysis.

\begin{definition}
\label{def:bad}
For every $\ell \in \{1,2,\ldots,|S_0|\}$, we define $i_\ell$ as the smallest $i$ for which $|S_i| = \ell$.
We refer to a given $\ell \in \{1,2,\ldots,\lfloor|S_0|/2\rfloor\}$ as bad if both $w_{\start}(B_{\start}) \leq 8\ell$ and $w_{\test}(B_{\test}) > 4\ell$ and otherwise we refer to $\ell$ as good.
\end{definition}
Note that $i_\ell$ is well-defined in the sense that there has to exist at least one $i$ with $|S_i| = \ell$ for every $\ell \in \{1,2,\ldots,|S_0|\}$. This follows from $|S_{i+1}| \geq |S_i| - 1$ for every $i \in \{1,2,\ldots,k-1\}$ and $|S_{k-1}| \leq 1$.
\begin{lemma}
\label{lem:approx}
Let $\ell_{max}$ be defined as the largest $\ell \in \{1,2,\ldots,\lfloor|S_0|/2\rfloor\}$ such that $\ell$ is bad, if there exists such an $\ell$, and otherwise let $\ell_{max} = 1$. Then, for every $i \in \{0,1,\ldots,k-1\}$, we have $\frac{w_i(E_i)}{k-i} \leq 90\ell_{max}$.
\end{lemma}
\begin{proof}
We first prove by induction that $w_i(B_i) \leq \max(4|S_i|,8\ell_{max})$ for every $i \in \{0,1,\ldots,k-1\}$. As our base case, we consider any $i$ with $|S_i| \geq |S_0|/2$. Using that the average weight is $1$ at the beginning, we get $|S_0| \geq k/2$ by Markov's inequality and therefore $w_i(B_i) \leq k \leq 2|S_0| \leq 4|S_i|$. For our induction step, consider some arbitrary $i$ with $|S_i| < |S_0|/2$. Let $\ell := |S_i|$. First, we consider the case that $\ell_{max} \geq \ell$.  In particular, this implies $|S_{i-1}| \leq |S_i| + 1 \leq \ell + 1 \leq \ell_{max} + 1$ and therefore we get by induction that
\[w_i(B_i) \leq w_{i-1}(B_{i-1}) \leq \max(4|S_{i-1}|,8\ell_{max}) \leq \max(4(\ell_{max} + 1),8\ell_{max}) \leq 8\ell_{max}.\]
Thus, it suffices to consider the case that $\ell > \ell_{max}$, which in particular implies that $\ell$ is good. We have $\start < \test \leq i$ (since $\ell \leq |S_0|/2 \leq i$) and therefore we can assume by induction that $w_{\start}(B_{\start}) \leq \max(4(2\ell),8\ell_{max}) = 8\ell$. As $\ell$ is good, this implies that $w_{\test}(B_{\test}) \leq 4\ell$ and therefore $w_i(B_i)\leq w_{\test}(B_{\test}) \leq 4\ell = 4|S_i|$. This finishes the induction and thus we indeed have $w_i(B_i) \leq \max(4|S_i|,8\ell_{max})$ for every $i \in \{0,1,\ldots,k-1\}$. Therefore,

\[\frac{w_i(E_i)}{k-i} \leq \frac{w_i(E_i)}{|S_i| + |M_i| + |B_i|} \leq \frac{w_i(B_i)}{\max(|S_i|,1)} + \frac{80(|S_i| + |M_i|)}{|S_i| + |M_i|} \leq \max(4,8\ell_{max}) + 80 \leq 90 \ell_{max}.\]
\end{proof}
\begin{lemma}
\label{lem:bad}
Let $\ell \in \{1,2,\ldots,\lfloor|S_0|/2\rfloor\}$. Then, $\ell$ is bad with probability at most $e^{-\frac{\ell}{40}}$.
\end{lemma}
For the proof of \cref{lem:bad}, we need the following Chernoff-bound variant.

\begin{lemma}[Chernoff bound]
\label{lem:chernoff}
Let $X_1, \dots, X_\ell$ be independent Bernoulli-distributed random variables, each equal to one with probability $p$. Then,
\[
\P\left(  \sum_{i = 1}^\ell X_i < \frac{p \ell}{2} \right) \leq e^{-p\ell/8}. 
\]
\end{lemma}
\begin{proof}[Proof of \cref{lem:bad}]
Throughout the proof, we assume that $w_{\start}(B_{\start}) \leq 8\ell$. In particular, 
\[|B_{\start}| \leq \frac{w_{\start}(B_{\start})}{80} \leq \frac{\ell}{10}.\]
Below, we will define for every $j \in \{1,2,\ldots,\ell\}$ an indicator variable $X_j$ in such a way that

\begin{enumerate}
    \item $\E[X_j|X_1,X_2,\ldots,X_{j-1}] \geq \frac{1}{5}$ for every $j \in \{1,2,\ldots,\ell\}$ and 
    \item if $X := \sum_{j=1}^{\ell} X_j \geq \frac{\ell}{10}$, then $w_{\test}(B_{\test}) \leq 4\ell$.
\end{enumerate}
The first property implies that $X$ stochastically dominates a random variable $X'$ which is the sum of $\ell$ independent Bernoulli-distributed random variables, each equal to one with probability $1/5$. Thus, using \cref{lem:chernoff}, we get
\[\P\left[X < \frac{\ell}{10}\right] \leq \P\left[X' < \frac{\ell}{10}\right] \leq e^{-\frac{\ell}{40}}.\]
Thus, we can now use the second property to deduce that $\ell$ is bad with probability at most $e^{-\frac{\ell}{40}}$. It thus remains to define the random variables and show that they indeed satisfy the two properties.
To that end, fix some $j \in \{1,2,\ldots,\ell\}$. We define $i'_j$ as the smallest $i \in \{\start,\start+1,\ldots,\test - 1\}$ with $|S_i| = 2\ell - j + 1$ and $e_{i+1} \notin M_i$. Note that there exists at least one such $i$ as there exists some $i$ with $|S_i| = 2\ell - j + 1$ and $|S_{i+1}| = 2\ell - j$, and for this $i$ it holds that $e_{i+1} \in S_i$ and therefore $e_{i+1} \notin M_i$. Note that it furthermore holds that $i'_1 < i'_2 < \ldots < i'_\ell$. We set $X_j = 1$ if $w_{i'_j}(B_{i'_j}) \leq 4 \ell$ or $e_{i'_j+1} \in B_{i'_j}$ and otherwise we set $X_j = 0$. We start by showing that the second property holds by proving the contrapositive. To that end, assume that $w_{\test}(B_{\test}) > 4\ell$. In particular, we have for every $j$ that $w_{i'_j}(B_{i'_j}) > 4 \ell$. Thus, if $X_j = 1$, we get $e_{i'_j+1} \in B_{i'_j}$ and therefore $|B_{i'_j + 1}| \leq |B_{i'_j}| - 1$. As $|B_{\start}| < \frac{\ell}{10}$, we therefore get that $X < \frac{\ell}{10}$, as needed.

It remains to show the first property. To that end, consider any $i$ and assume we have already sampled $e_1,\ldots,e_i$ in an arbitrary manner such that $|S_i| \leq 2\ell$ and $w_i(B_i) \geq 4\ell$. Then, conditioned on $e_{i+1} \notin M_i$, we get $e_{i+1} \in B_i$ with probability at least

\[\frac{D^{\eps}_i(B_i)}{D^{\eps}_i(B_i) + D^{\eps}_i(S_i)} \geq \frac{(1-\eps)w_i(B_i)}{(1-\eps)w_i(B_i) + (1+\eps)w_i(S_i)} \geq \frac{0.5 \cdot 4 \ell}{0.5 \cdot 4 \ell + 1.5 \cdot 2 \cdot 2 \ell} \geq \frac{1}{5}.\]

In particular, this directly implies $\E[X_j|X_1,X_2,\ldots,X_{j-1}] \geq \frac{1}{5}$ for every $j \in \{1,2,\ldots,\ell\}$.
 \end{proof}

Finally, we are ready to prove \cref{lem:main} by combining \cref{lem:approx,lem:bad}.
\begin{proof}[Proof of \cref{lem:main}]
Fix some $i \in \{0,1,\ldots,k-1\}$. Let $\ell_{max}$ be defined as in \cref{lem:approx}. \cref{lem:approx} gives that for every $\ell$ with $Pr[\ell_{max} = \ell] > 0$, we have
\[\E\left[\frac{\sum_{e \in E_i} w_i(e)}{k-i}| \ell_{max} = \ell\right] \leq 90 \ell.\]
Moreover, for $\ell >1$ , we can use \cref{lem:bad} to deduce that $\P[\ell_{max} = \ell] \leq \P[\text{$\ell$ is bad}] \leq e^{-\frac{\ell}{40}}$.
Therefore,
\[\E\left[\frac{\sum_{e \in E_i} w_i(e)}{k-i}\right] \leq \sum_{\ell = 1}^\infty 90 \ell \cdot e^{-\frac{\ell - 1}{40}} = O(1).\]
\end{proof}

\bibliographystyle{alpha}
\bibliography{ref}

\end{document}